\documentclass[sigconf]{acmart}
\setcopyright{none}
\renewcommand\footnotetextcopyrightpermission[1]{}
\usepackage{hyperref}
\usepackage{url}
\usepackage{dsfont}
\usepackage{subcaption}
\usepackage{booktabs}
\usepackage{multicol,multirow}
\usepackage{wrapfig}
\usepackage{graphicx}
\usepackage{makecell}

\usepackage{amsmath,amsfonts,bm}

\def\eqref#1{equation~\ref{#1}}

\def\1{\bm{1}}

\def\vzero{{\bm{0}}}

\def\va{{\bm{a}}}

\def\vu{{\bm{u}}}
\def\vv{{\bm{v}}}
\def\vw{{\bm{w}}}
\def\vx{{\bm{x}}}
\def\vy{{\bm{y}}}

\def\mW{{\bm{W}}}

\DeclareMathAlphabet{\mathsfit}{\encodingdefault}{\sfdefault}{m}{sl}
\SetMathAlphabet{\mathsfit}{bold}{\encodingdefault}{\sfdefault}{bx}{n}

\newcommand{\softmax}{\mathrm{softmax}}

\DeclareMathOperator*{\argmax}{arg\,max}

\usepackage{enumitem}
\newcommand\supp{\mathrm{supp}}
\newcommand\KeepTopK{\operatorname{KeepTopK}}

\newcommand{\mtd}{\textsc{MoS}}

\newtheorem{phe}{Phenomenon}
\newtheorem{prop}{Proposition}

\AtBeginDocument{%
  }

\copyrightyear{2026}
\acmYear{2026}
\setcopyright{cc}
\setcctype{by}
\acmConference[WWW '26]{Proceedings of the ACM Web Conference 2026}{April 13--17, 2026}{Dubai, United Arab Emirates}
\acmBooktitle{Proceedings of the ACM Web Conference 2026 (WWW '26), April 13--17, 2026, Dubai, United Arab Emirates}
\acmPrice{}
\acmDOI{10.1145/3774904.3792483}
\acmISBN{979-8-4007-2307-0/2026/04}

\begin{document}

\title{Mixture of Sequence: Theme-Aware Mixture-of-Experts for Long-Sequence Recommendation}


\settopmatter{authorsperrow=5}



\author{Xiao Lin}
\affiliation{%
  \institution{UIUC}
  \country{United States}
}
\email{xiaol13@illinois.edu}

\author{Zhicheng Tang}
\affiliation{%
  \institution{Meta}
  \country{United States}
}
\email{roberttang@meta.com}

\author{Weilin Cong, Mengyue Hang}
\affiliation{%
  \institution{Meta}
  \country{United States}
}

\author{Kai Wang}
\affiliation{%
  \institution{Meta}
  \country{United States}
}
\email{wangkai@meta.com}

\author{Yajuan Wang}
\affiliation{%
  \institution{Meta}
  \country{United States}
}
\email{yajuanwang@meta.com}

\author{Zhichen	Zeng, Ting-Wei Li, Hyunsik Yoo}
\affiliation{%
  \institution{UIUC}
  \country{United States}
}

\author{Zhining Liu, Xuying Ning, Ruizhong Qiu}
\affiliation{%
  \institution{UIUC}
  \country{United States}
}

\author{Wen-Yen	Chen, Shuo Chang, Rong Jin}
\affiliation{%
  \institution{Meta}
  \country{United States}
}

\author{Huayu Li}
\affiliation{%
  \institution{Meta}
  \country{United States}
}
\email{huayuli@meta.com}

\author{Hanghang Tong}
\affiliation{%
  \institution{UIUC}
  \country{United States}
}
\email{htong@illinois.edu}

\renewcommand{\shortauthors}{Xiao Lin et al.}

\begin{abstract}
Sequential recommendation has emerged as a rapidly growing research area in click-through rate prediction due to its ability to capture dynamic user interests from historical interaction sequences. A key challenge, however, lies in modeling long sequences, where users often exhibit pronounced interest shifts, thereby introducing substantial irrelevant or even misleading information into the prediction process. Our empirical analysis corroborates this challenge and further uncovers a recurring behavioral pattern in long sequences, which we term the \textit{session hopping} phenomenon: while user interests remain stable within a short temporal span, referred to as a \textit{session}, they often exhibit drastic shifts across sessions and may reappear after multiple sessions.
To address this challenge, we propose the Mixture of Sequence (MoS) framework, a model-agnostic MoE approach that achieves accurate predictions by extracting theme-specific and multi-scale subsequences from noisy raw user sequences. First, MoS employs a theme-aware routing mechanism to adaptively learn the latent themes of user sequences and organizes these sequences into multiple coherent subsequences. Each subsequence contains only sessions aligned with a specific theme, thereby effectively filtering out irrelevant or even misleading information introduced by user interest shifts in session hopping. In addition, to alleviate potential information loss caused by subsequence extraction, we introduce a multi-scale fusion mechanism, which leverages three types of experts to capture global sequence characteristics, short-term user behaviors, and theme-specific semantic patterns. Together, these two mechanisms endow MoS with the ability to deliver accurate recommendations from multi-faceted and multi-scale perspectives. Experimental results demonstrate that MoS consistently improves the performance of long-sequence recommendation models while introducing fewer FLOPs compared with other MoE counterparts, providing strong evidence of its excellent balance between utility and efficiency. The code is available at \url{https://github.com/xiaolin-cs/MoS}.

\end{abstract}

\vspace{-2mm}
\vspace{-2mm}


\maketitle

\vspace{-2mm}
\section{Introduction}

Click-through rate (CTR) prediction \cite{ctr1, ctr2, ctr3} serves as a cornerstone in online advertising \cite{wang2020survey} and recommender systems \cite{ctr3,ctr1,ctr2,jing2022coin,jing2024sterling,li2022unsupervised,li2024large}, enabling users to efficiently discover items of interest from an ever-growing pool of products \cite{product}, videos \cite{video1,video2}, and applications \cite{content}. In recent years, modeling user behavior sequences has proven to be a highly effective strategy for enhancing CTR performance, as such sequences capture rich temporal dynamics that reflect the evolution of user interests. This advantage has spurred substantial research interest in sequence recommendation \cite{seqrec1, seqrec2,seqrec3,seqrec4}. Furthermore, extending the sequence length is particularly appealing, as it allows models to exploit long-term behavioral patterns. Motivated by this, recent studies \cite{chang2023twin,xia2023transact,sdim,chai2025longer} have begun to investigate the feasibility of modeling very long sequences and have reported promising improvements.

Despite the advantages of leveraging long interaction sequences for CTR prediction, they introduce significant challenges. As user interests often undergo substantial shifts within long sequences \cite{shift1,shift2,shift3}, many interactions become irrelevant to the current prediction or even inject misleading signals. For instance, on an e-commerce platform \cite{e-commerce}, a user may simultaneously exhibit diverse interests, such as browsing both electronic devices and athletic equipment. An illustration is provided in Figure. \ref{fig:session hopping}. Ideally, when recommending an electronic product, the model should primarily focus on the subsequence associated with electronics while down-weighting interactions w.r.t. unrelated interests, such as athletic equipment. However, if the model indiscriminately incorporates the entire sequence, misleading interactions can dominate the representation and ultimately lead to erroneous predictions \cite{noisy1,noisy2}. This issue motivates us to raise a critical question for long-sequence modeling.

\textit{How could a recommendation model, by focusing on strategically chosen subsets of a long sequence, gain complementary perspectives 
that lead to more accurate predictions?}

\begin{figure}
    \centering
    \includegraphics[width=\linewidth]{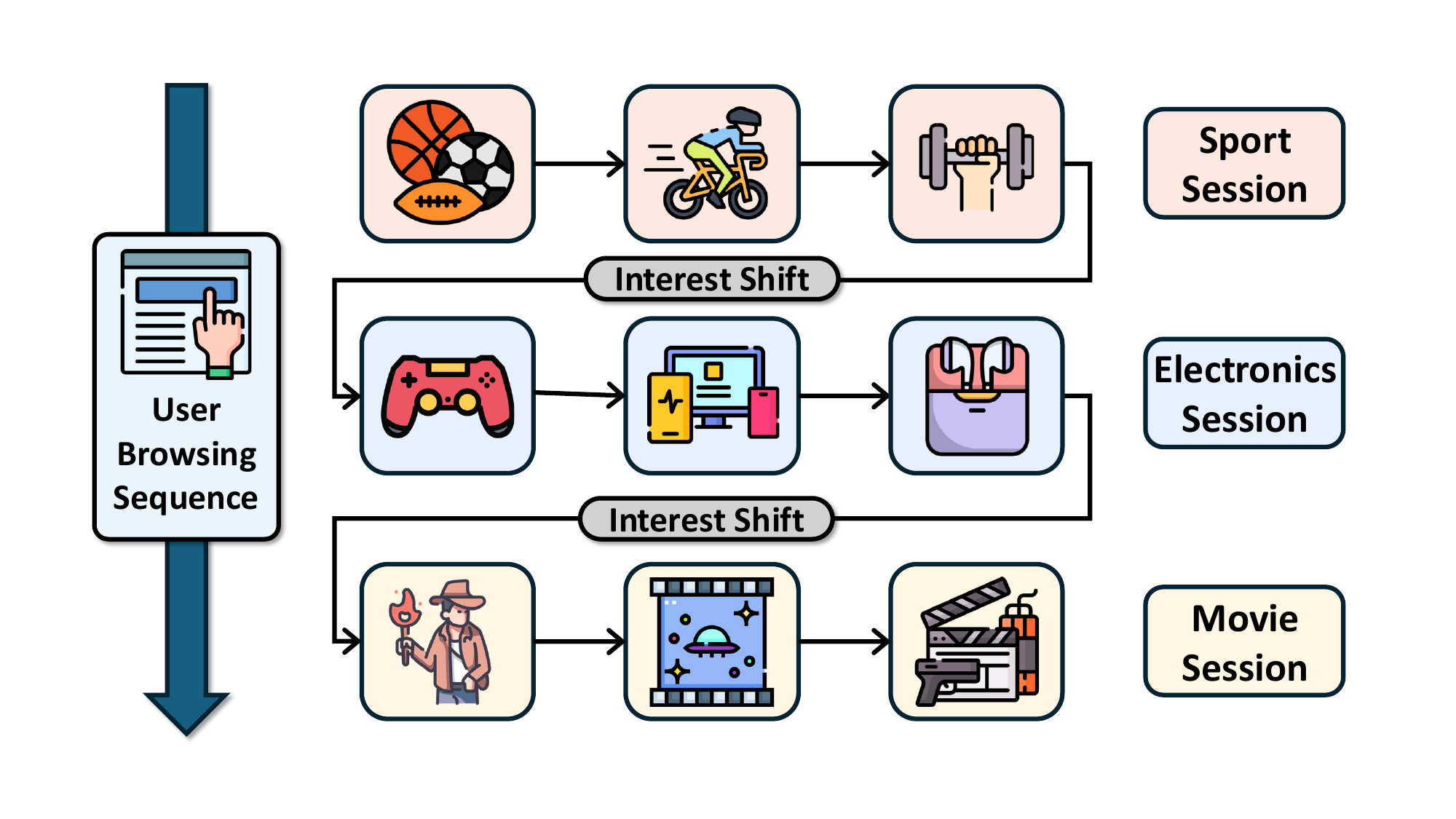}
    \vspace{-10mm}
    \caption{An illustration of user interests. The user undergoes two interest shifts, resulting in three clearly distinct sessions. Within each session, the browsing history reflects a consistent user interest.}
    \label{fig:session hopping}
    \vspace{-5mm}
\end{figure}

To address this question, we investigate the distinctive behavioral patterns embedded in long sequential data. Through a comprehensive empirical analysis, we identify a phenomenon we term \textit{session hopping}: user interests remain highly consistent within a short temporal span (a “\textit{session}”), yet may shift abruptly across adjacent sessions and reappear after several sessions. This phenomenon clearly demonstrates the sudden and recurring transitions in user preferences, highlighting the substantial challenges it poses for long-sequence recommendation. Existing methods often rely on attention \cite{xia2023transact,sdim,chang2023twin}, selective copying \cite{mamba2,liu2024mamba4rec} or similar mechanisms \cite{rnn1} to capture long-range dependencies, but these mechanisms are designed as continuous function approximators for smooth dynamics. As such, they struggle to fit the discontinuous signals from abrupt interest shifts and hence continue to exploit information from sessions prior to the shift. As illustrated in Figure~\ref{fig:session hopping}, once an interest shift occurs, the items chosen by the user differ substantially from those in earlier sessions. Consequently, prior session information not only fails to provide useful guidance but may also mislead the prediction.
 
To address the challenges introduced by the session hopping phenomenon, we propose the Mixture of Sequence (\mtd{}), a model-agnostic MoE framework that leverages subsequence extraction to selectively utilize informative sessions for accurate prediction. Intuitively, since similar sessions tend to reappear, these sessions can be grouped into subsequences that reflect specific themes of user behaviors (e.g., sports). By assigning a dedicated expert to each subsequence, \mtd{} enables every expert to fully exploit informative sessions in the theme-aware subsequence while effectively filtering out misleading signals from other themes, thereby mitigating the negative effects of abrupt interest shifts. As a result, \mtd{} serves as a versatile plug-in design that empowers existing sequential models to overcome the challenges posed by session hopping. Specifically, \mtd{} consists of two key components: theme-aware routing and a multi-scale fusion mechanism.
(1) The theme-aware routing mechanism introduces a router equipped with a codebook that defines user-behavior themes, allowing \mtd{} to adaptively extract highly coherent, theme-specific subsequences from long interaction histories. (2) The multi-scale fusion mechanism mitigates the potential disruption of semantic and temporal continuity caused by subsequence extraction. It incorporates three types of experts that respectively capture global characteristics of the sequence, theme-specific patterns within subsequences, and short-term contiguous user behaviors. Together, these two components enable \mtd{} to model user dynamics from multi-thematic and multi-scale perspectives, leading to more accurate sequential recommendations.

In summary, our main contributions are as follows:
\begin{itemize}[leftmargin=2em, labelsep=1em]
    \item \textbf{Observation.} We empirically demonstrate the existence of the \textit{session hopping} phenomenon on real-world data, which manifests as both abrupt shifts and reappearance of user interests, thereby highlighting the inherent challenges of modeling long-sequence recommendation.
    \item \textbf{Algorithm.} As a model-agnostic MoE method, \mtd{}  introduces a router that perform adaptive theme-aware subsequence extraction, and leverages three types of experts to jointly capture global, short-term, and semantic characteristics. These designs enable accurate recommendations even in the presence of highly noisy long-sequence data.
    \item \textbf{Experiment.} We conduct extensive experiments on three real-world datasets. The results show that \mtd{} consistently improves performance across four sequence recommendation backbones, yielding an average gain of 0.68\% in AUC and 0.72\% in GAUC. \mtd{} surpasses all other MoE methods while requiring fewer FLOPs, demonstrating an excellent balance between utility and efficiency.
\end{itemize}

\vspace{-2mm}
\section{Preliminary}
\subsection{Sequential Recommendation}
\vspace{-1mm}

Let $\mathcal{U}$ and $\mathcal{V}$ denote the sets of users and items, respectively, with $u \in \mathcal{U}$ and $v \in \mathcal{V}$ representing a specific user and item. We use $\vert \mathcal{U} \vert$ and $\vert \mathcal{V} \vert$ to denote the sizes of these sets. For each user $u$, we define the interaction history as a chronologically ordered sequence $\mathcal{S}_u = (v_{1, u}, v_{2,u}, \dots, v_{t, u})$ where $v_{i,u}$ is the $i$-th interacted item and $t$ is the sequence length for user $u$. For notational simplicity, we use $v_{i}$ to denote $v_{i,u}$ and $\mathcal{S}$ to denote $\mathcal{S}_u$ in the following. Given a user interaction sequence $\mathcal{S}$, the objective of sequential recommendation is to predict the next item the user will interact with at time step $t+1$. Formally, the task can be expressed as:
\begin{equation}
    v_{t+1}^* = \argmax_{v \in \mathcal{V}}P(v \vert \mathcal{S}; \Theta)
\end{equation}
where $\Theta$ denotes the model parameters.
\vspace{-1mm}
\subsection{Sparse Mixture of Experts}

For a sparse MoE model with $n$ experts $\{F^{(1)}, \ldots, F^{(n)}\}$, a learnable sparse router $G(\cdot)$ is leveraged to decide which experts are activated for a given item embedding $\vx$. The model prediction is a weighted combination of expert outputs, i.e., $\vy = \sum_{i=1}^n G(\vx)[i] F^{(n)}(\vx)$ where $G(\vx)[i]$ denotes the $i$-th element of $G(\vx)$, i.e., the routing weight assigned to expert $F_i$. In practice, $G(\cdot)$ is parameterized by a scoring network $H(\cdot)$ that produces relevance scores of experts. To enforce sparsity, only the top-$k$ experts with the highest scores are selected:
\begin{equation}\label{eq:smoe}
    G(\vx) = \softmax (\KeepTopK(H(\vx), k))
\end{equation}
where $\KeepTopK(\vv, k)$ preserves the $k$ largest elements of $\vv$ and replaces the rest with $-\infty$.
\section{Session Hopping}

\begin{figure}[t]
  \centering
  \begin{minipage}[c]{0.63\linewidth}
    \includegraphics[width=\linewidth]{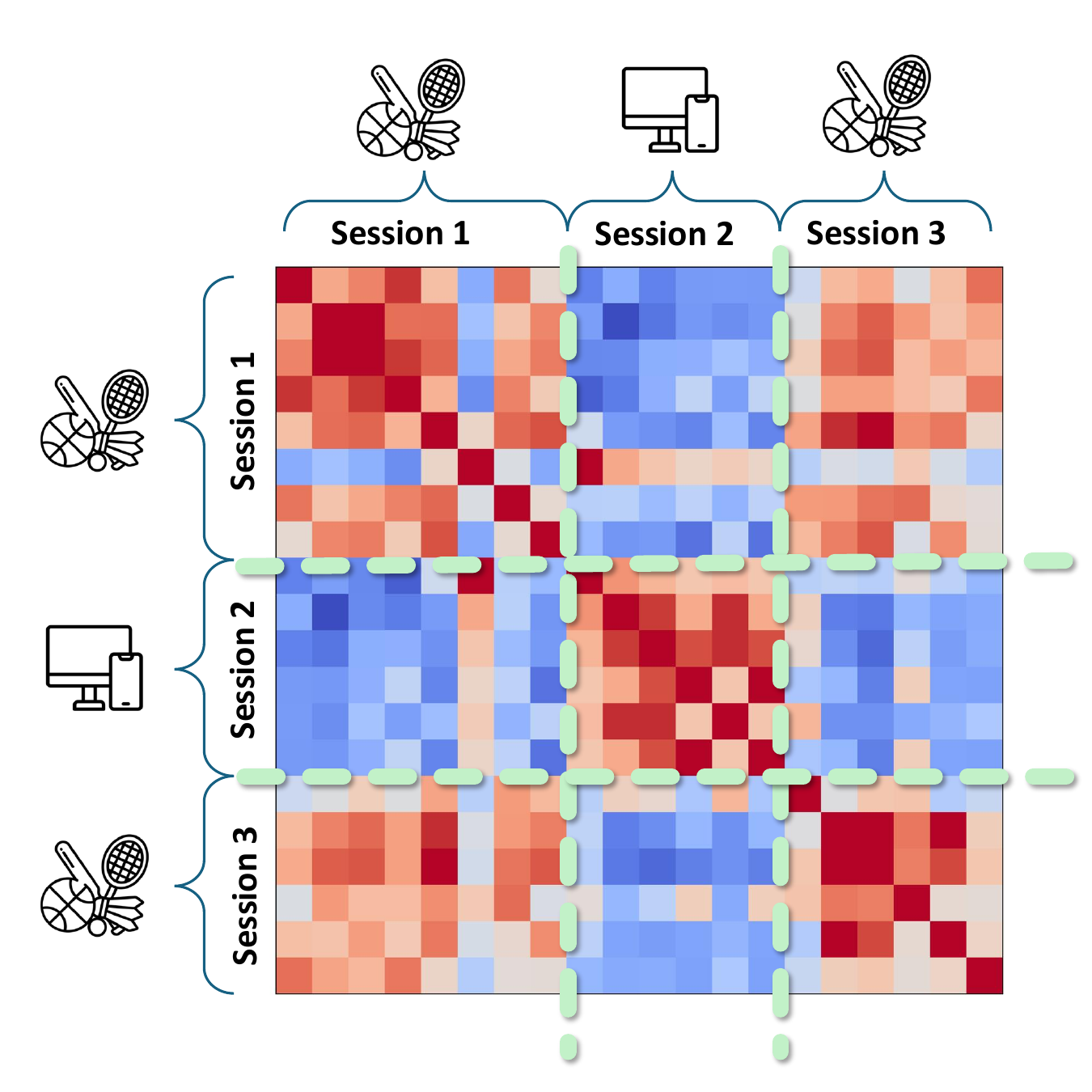}
  \end{minipage}\hfill
  \begin{minipage}[c]{0.35\linewidth}
    \captionof{figure}{Heatmap of the self-similarity matrix for a representative user transaction history.
    Red indicates \textcolor{red}{high similarity}, blue indicates \textcolor{blue}{low similarity}, and green lines denote session boundaries.}
    \label{fig:autocorrelation}
  \end{minipage}
  \vspace{-4mm}
\end{figure}

To better understand the challenges of long-sequence modeling, we conduct a case study on real-world datasets. Specifically, we compute the self-similarity matrix of user interaction history using cosine similarity.
A representative heatmap of user behavior is shown in Figure~\ref{fig:autocorrelation}, with more examples provided in Appendix \ref{appdix:session_examples}.
From Figure~\ref{fig:autocorrelation}, we can clearly identify a distinctive user behavior pattern, which we term the \textit{session hopping} phenomenon. This phenomenon is characterized by three properties: (1) \textbf{Stability.} User interests remain highly consistent and stable within a short temporal span (a session), as indicated by the red area of Figure~\ref{fig:autocorrelation}; (2) \textbf{Discontinuity.} User interests undergo abrupt and discontinuous shifts across sessions, reflected in the sharp red–blue boundary in Figure~\ref{fig:autocorrelation}; (3) \textbf{Reappearance.} User interests may reappear after several sessions, as indicated by the strong similarity between the first and third sessions, both highlighted in red. This pattern closely aligns with everyday user behavior. For instance, when purchasing a computer, users often buy related accessories such as a mouse or keyboard within a short period of time. Once these needs are fulfilled, however, their interest in electronics typically drops, and attention shifts abruptly to unrelated categories. After some time, users may once again develop interest in electronics, such as upgrading to a better mouse. A clear illustration is provided in Figure \ref{fig:session hopping}. In summary, we formalize this phenomenon as follows:

\begin{phe}
    (\textbf{Session Hopping}) User interests remain highly consistent within a session but differ substantially across adjacent sessions. Moreover, similar user interests may reappear across non-adjacent sessions.
\end{phe}

This phenomenon illustrates both the opportunities and challenges of long-sequence recommendation. On the one hand, the reappearance property of session hopping suggests that certain user behavior patterns tend to recur over time. These recurring patterns provide valuable cues for predicting the next item, and compared with short-sequence settings, long-sequence data offer richer opportunities to observe such recurring sessions, thereby holding substantial potential for improving predictive accuracy. On the other hand, the discontinuity inherent in session hopping introduces significant modeling challenges. Mainstream sequential recommendation models, such as RNN-based \cite{rnn1}, Mamba-based \cite{mamba2,liu2024mamba4rec}, and Transformer-based architectures \cite{xia2023transact,sdim,chang2023twin}, rely on continuous functions during forward propagation and thus lack mechanisms explicitly designed to capture abrupt, non-continuous shifts in user interests. When fitting such discontinuous signals with models designed smooth dynamics \cite{castin2023smooth,nishikawa2024state}, approximation errors naturally arise. In recommendation scenarios, when a sudden shift in user interest occurs, the content of preceding sessions may differ drastically from the current intent. As a result, earlier sessions not only fail to provide useful information but may even mislead the prediction. Taking Transformer-based models as an example, although attention mechanisms can partially down-weight irrelevant interactions, it is practically impossible for them to assign zero attention weights to all unrelated sessions preceding a shift. Consequently, the model is highly likely to attend to misleading contexts, thereby incorporating harmful information into recommendation decisions. In summary, current sequential recommendation models face a fundamental trade-off: leveraging as much informative historical context as possible while mitigating the adverse influence of misleading information in long sequences.

\vspace{-2mm}
\section{Methodology}

\begin{figure*}
    \centering
    \includegraphics[width=0.98\linewidth]{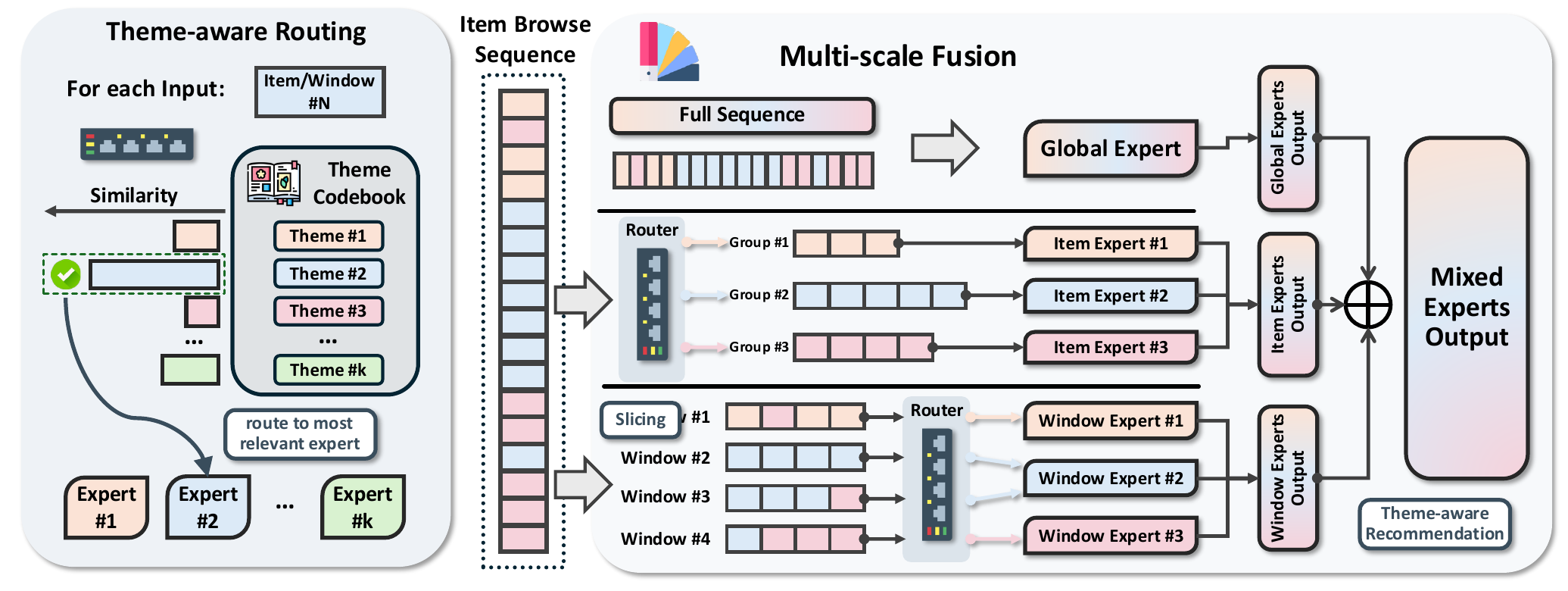}
    \vspace{-2mm}
    \caption{The pipeline of \mtd{}. The left panel illustrates theme-aware routing, which assigns inputs to experts according to theme vectors in the codebook. The right panel shows multi-scale fusion, which models user behaviors by extracting subsequences at different granularities from the full sequence.}
    \label{fig:pipeline}
    \vspace{-2mm}
\end{figure*}

To address the problem introduced by session hopping, we propose \mtd{} framework, which adaptively selects highly coherent subsequences at multiple scales to provide diverse perspectives for accurate prediction. \mtd{} is built upon two key mechanisms: the theme-aware routing mechanism in Section \ref{sec:theme_routing} and the multi-scale fusion mechanism in Section \ref{sec:multi-scale}. Furthermore, we discuss the training paradigm of \mtd{} in Section~\ref{sec:training} and analyze its computational efficiency in Section~\ref{sec:efficiency}. The pipeline of \mtd{} is shown in Figure \ref{fig:pipeline}.

\vspace{-1mm}
\subsection{Theme-aware Routing}\label{sec:theme_routing}

The session hopping phenomenon introduces substantial difficulties for long-sequence recommendation, and existing approaches struggle to precisely capture such user behavior patterns. However, we propose that a simple yet elegant solution is \textit{subsequence extraction}. Since user interests remain relatively stable within a session and similar sessions often reappear, these sessions reflect a particular user interest or a combination of closely related interests. If we define such interests as a special “\textit{theme}” (e.g., sports), then these sessions can be regarded as a subsequence aligned with that theme. Furthermore, if a long user sequence can be ideally decomposed into multiple subsequences according to themes, then temporally adjacent but semantically irrelevant sessions are automatically separated into different sequences, thereby effectively reducing the disruptive impact of sudden interest shifts.

The idea of subsequence extraction naturally aligns with the MoE framework. By assigning individual experts to model each subsequence, every expert is able to capture long-range dependencies across sessions in a theme-aware subsequence while effectively ignoring irrelevant information.

\textbf{Model architecture.} 
Based on the above intuition, we employ a router to adaptively extract multiple subsequences, each of which is then processed by a dedicated expert. Formally, given a router $G(\cdot)$ and a user sequence $\mathcal{S}$, the subsequence assigned to expert $i$ is defined as $\mathcal{S}^{(i)}:= (v | G(\vx_{v})[i]> 0, v \in \mathcal{S})$ with $\vx_v$ being the embedding of the item $v$. This dispatch process needs to satisfy two key properties: (1) \textbf{Sparsity.} To avoid the influence of misleading sessions, each expert should only have access to the corresponding subsequence instead of the entire sequence. Hence, a sparse MoE is leveraged by \mtd{} rather than a dense MoE. (2) \textbf{Cohesion.} The subsequence assigned to an expert should be semantically coherent and aligned with a consistent theme, which imposes strong requirements on the router design.
However, empirical evidence shows that conventional routers (e.g., MLPs or linear layers) fail to fulfill this requirement, as they tend to distribute items uniformly across experts rather than making theme-aware assignments. To address this limitation, we propose the theme-aware routing mechanism, which maintains a theme codebook $\mW \in \mathbb{R}^{n\times D}$ with $D$ being the dimension of the codebook and $n$ being the number of experts. Here, the $i$-th row of the codebook describes the theme feature associated with the $i$-th expert. Given an item embedding $\vx$, MoS first projects $\vx$ into the theme space via a learnable MLP $h(\cdot)$, and then computes its cosine similarity with each entry of the codebook to obtain the relevance scores. Formally, the scoring network $H(\cdot)$ of the theme-aware router can be expressed as:
\begin{equation}\label{eq:theme_routing}
H(\vx) = \left(\Vert h(\vx) \Vert_2\sqrt{\operatorname{diag}(\mW \mW^T)} \right)^{-1} \mW h(\vx)
\end{equation}
where $\operatorname{diag}(\cdot)$ preserves only diagonal elements while setting other elements to zero. Mathematically, \eqref{eq:theme_routing} could easily ensure the following proposition with its proof provided in Appendix \ref{appdix:dispatch_proof}.
\begin{prop}\label{prop:theme_dispatch}
    (\textbf{Theme-aware Dispatch}) A subsequence with high internal similarity will be consistently routed to the same expert.
\end{prop}

\textbf{Optimization.} Although Proposition \ref{prop:theme_dispatch} theoretically guarantees the feasibility of the router, in practice the codebook $\mW$ cannot be trained by gradient updates under joint optimization. This is because the codebook and the learnable MLP tend to collapse into a single MLP due to linearity during joint training, thereby undermining the intended role of the codebook. To address this problem, inspired by VQ-VAE \cite{vqvae}, we update the codebook using an exponential moving average (EMA):
\vspace{-1mm}
\begin{equation}
    \mW^{(t)}_i = \gamma \mW^{(t-1)}_i + (1-\gamma) \frac{\sum_{\vx \in \mathcal{B}} \mathds{1}\left(G(\vx)[i] > 0\right) h(\vx) }{\sum_{\vx \in \mathcal{B}} \mathds{1}\left(G(\vx)[i] > 0\right)} 
\end{equation}
where $\mathcal{B}$ is a batch sampled from the dataset. In this way, although the item embeddings evolve throughout training, each entry of the codebook adaptively tracks the centroid of item embeddings under its corresponding theme, thereby maintaining stable and meaningful theme representations.

\textbf{Initialization.} A notable limitation of EMA is that if two entries are initialized too closely, they tend to collapse toward each other during the updating process, thereby reducing the diversity of the codebook. To mitigate this issue, it is essential to adopt an initialization strategy that ensures sufficient differentiation among themes. Specifically, we first estimate the distribution of item embeddings using a pretrained embedding layer, and then apply $k$-means clustering to obtain $k$ cluster centroids, which are subsequently used as the initial values of rows in the codebook.

\subsection{Multi-scale Fusion}\label{sec:multi-scale}

A notable limitation of theme-aware routing is that subsequence extraction disrupts temporal dependencies. As a result, information about users’ evolving preferences is largely ignored, and the model becomes biased toward recommending items that shows high semantic similarity rather than capturing dynamic behavioral shifts. This issue arises because the extraction process focuses solely on semantic coherence while neglecting temporal continuity. To remedy this, we introduce a multi-scale fusion mechanism that leverages a global expert, item experts, and window experts to capture temporal information at different granularities.

\textbf{Global Expert.} The objective of the global expert is to effectively extract global features from the sequence, thereby capturing users’ globally drifted interests and modeling their evolving preferences. To this end, the global expert directly takes the entire historical sequence as input, without relying on subsequence extraction. When only the global expert is activated, the model naturally degenerates into the non-MoE backbone. Concretely, given any model with a repeated block structure, if we treat one block $F(\cdot)$ as an expert, the output of the global expert $\vy_g$ can be expressed as:
\vspace{-1mm}
\begin{equation}
    \vy_G = F \left(\mathcal{X}\left(\mathcal{S}\right)\right)
\end{equation}
where $\mathcal{S}:=(v_1,\dots, v_t)$ denotes a user historical sequence and $\mathcal{X}(\mathcal{S}):=(\vx_{v_1}, \dots, \vx_{v_t})$ denotes the corresponding embedding sequence with $\vx_{v_i}$ being the embedding of the item $v_i$.

\textbf{Item Experts.} The goal of the item experts is to deliver fine-grained and accurate recommendations by decomposing user behaviors into multiple latent themes and analyzing the theme-aware behavioral patterns. To this end, leveraging theme-aware routing mechanism, the item router $G_I(\cdot)$ first organizes the historical sequence $\mathcal{S}$ into a set of thematically coherent subsequences $\mathcal{S}_{I}^{(i)} := (v | G_I(\vx_{v})[i]> 0, v \in \mathcal{S})$. Then each item expert $F_I^{(i)}(\cdot)$ independently analyzes its subsequence to capture user behaviors from a specific theme. Finally, the outputs of all item experts are aggregated to form multi-faceted recommendation results. Mathematically, this process could be expressed as:
\begin{equation}
    \vy_I = \sum_{i=1}^n G_I \left(\vx_{v_t}\right)\left[i\right] \cdot F_I^{(i)}\left(\mathcal{X}\left(\mathcal{S}_I^{(i)}\right)\right)
\end{equation}
where $\vx_{v_t}$ represents the embedding of the last item interacted with by the user $u$ and $\vy_I$ is the aggregated output by item experts.

\textbf{Window Experts.} Positioned between the global expert and the item experts, window experts jointly account for both semantic coherence and temporal continuity. These experts aim to capture short-term dynamics of user preferences while preserving theme-specific behavioral patterns. Concretely, we first transform the user’s historical item sequence into a historical window sequence using a sliding window of size $L$ and stride $s$. Mathematically, each window is defined as $w_m = (v_{s \cdot m+1}, \dots, v_{s\cdot m+L})$, and the corresponding window embeddings is then computed as the average of the embeddings of the items it contains, i.e., $\vx_{w_m} = \frac{1}{L} \sum_{j=s\cdot m+1}^{s\cdot m+L} \vx_{v_j}$. Then, given this transformed window sequence $\mathcal{S}_W = (w_1, \dots, w_{\Tilde{t}})$, similarly, a window router $G^W(\cdot)$ applies the theme-aware routing mechanism to extract several theme-aware windows subsequences, $\mathcal{S}_W^{(i)} := (w | G_W(\vx_{w})[i]> 0, w \in \mathcal{S}_W)$, with each subsequence $\mathcal{S}_W^{(i)}$ assigned to a window expert $F_W^{(i)}(\cdot)$. Finally, their outputs are aggregated to obtain the final prediction:
\begin{equation}
    \vy_W = \sum_{i=1}^n G_W\left(\vx_{w_{\tilde{t}}}\right)\left[i\right] \cdot F_W^{(i)}\left(\mathcal{X}\left(\mathcal{S}_W^{(i)}\right) \right)
\end{equation}

\textbf{Multi-scale Fusion.} 
By integrating the outputs of the global, item, and window experts, MoS employs a multi-scale fusion strategy that effectively balances semantic coherence and temporal continuity in recommendation. When applied to sequential data exhibiting the session hopping phenomenon, MoS provides an elegant solution that enables accurate recommendations across multiple thematic perspectives and temporal granularities. Concretely, the final output of MoS is obtained as a weighted combination of the three expert groups:
\begin{equation}
\vy = (1 - \alpha_I - \alpha_W)\vy_{G} + \alpha_I \vy_{I} + \alpha_W \vy_{W},
\end{equation}
where $\alpha_I$ and $\alpha_W$ are two hyperparameters that control the relative importance of item and window experts in the final prediction.

\subsection{Training Paradigm} \label{sec:training}

Beyond model architecture design, a dedicated training paradigm is essential to ensure stable and effective optimization, particularly given the high training complexity of MoE models. To this end, we adopt a staged progressive training paradigm, which consists of three phases: backbone warm-up, theme-aware expert warm-up, and joint optimization.

\textbf{Backbone Warm-up.}
An insightful observation in our model design is that when only the global expert is activated, the entire architecture degenerates into the backbone model without MoE. This implies that training solely the global expert substantially reduces the optimization difficulty, while guaranteeing performance at least comparable to the backbone. Motivated by this, in the backbone warm-up stage, \mtd{} trains only the backbone components, including the global expert, embedding layer, and classifier head. Meanwhile, all the output from theme-aware experts are ignored, i.e., $\alpha_I = \alpha_W =0$.

\textbf{Theme-aware Expert Warm-up.}
In this stage, we update only the parameters of the item experts and window experts, while freezing the embedding layer and classifier head, and ignoring the output of the global expert, i.e., $\alpha_I = \alpha_W = 0.5$. Since the embedding layer and classifier head have already been well warmed up, this stage forces the item and window experts to align their outputs with those of the global expert, hence providing a strong initialization for the theme-aware experts.

\textbf{Joint Optimization.} In the final stage, we jointly fine-tune the entire model, updating all parameters based on the strong initialization provided by the previous phases. Here, $\alpha_I$ and $\alpha_W$ are treated as tunable hyperparameters to balance the contributions of item and window experts. This stage ensures global consistency across all components and maximizes the overall performance of the model.

\vspace{-1mm}
\subsection{Efficiency Analysis} \label{sec:efficiency}

As an MoE-based method, \mtd{} not only achieves strong effectiveness but also introduces only a marginal increase in computational overhead, thereby ensuring efficiency. Given the prevalence of Transformer architectures, we adopt Transformer blocks as an example. Suppose \mtd{} activates the top-$k_1$ experts out of $n_1$ item experts and the top-$k_2$ experts out of $n_2$ window experts. In this case, each item expert is responsible for approximately $O(\frac{N}{n_1})$ data points, while each window expert processes about $O(\frac{N}{s \cdot n_2})$ data points due to the slicing window, where $N$ is the sequence length and $s$ is the stride. Consequently, compared with the backbone, the additional time complexity introduced by MoS is $O(\frac{k_1^2 N^2}{n_1} + \frac{k_2^2 N^2}{s^2\cdot n_2})$. Under the same conditions, this complexity is substantially lower than that of other MoE methods with a shared expert \cite{dai2024deepseekmoe, lepikhin2020gshard,hazimeh2021dselect,zhou2022expertchoice}, whose complexity is $O(\frac{(k_1 + k_2)^2 N^2}{(n_1 + n_2)})$. For instance, in our experimental setup with $k_1 = k_2$, $n_1 = n_2$, and $s=4$, the theoretical time complexity of \mtd{} is only about 53\% of that of conventional MoE counterparts, highlighting its high efficiency.

\vspace{-2mm}
\section{Experiments}
\subsection{Experiment Settings} \label{subsec:exp_setting}

\paragraph{Dataset descriptions.}
We evaluate \mtd{} on three real-world datasets: MicroVideo~\cite{microvideo}, KuaiVideo~\cite{kuaivideo}, and EBNeRD-Small~\cite{ebnerd}. The MicroVideo dataset focuses on short-video recommendation with multimodal image embeddings in the entertainment domain, KuaiVideo originates from the Kuaishou competition and models  user–video interactions on online video platforms, and EBNeRD-Small targets personalized news recommendation in the digital publishing domain. Detailed dataset descriptions are provided in Appendix~\ref{appdx:dataset_desc}.

\paragraph{Evaluation metrics.} We evaluate model performance by ranking predictions over the entire item set without negative sampling. To measure utility, we adopt AUC and GAUC, while FLOPs are used to assess efficiency. Detailed definitions of these metrics are provided in Appendix~\ref{appdix:metric_desc}. For AUC and GAUC, higher values indicate better performance, whereas for FLOPs, lower values are preferred.

\paragraph{Models.} We demonstrate the superiority of MoS from two perspectives: routing strategy and multi-scale fusion. From the routing perspective, we compare MoS with three MoE baselines, namely GShard \cite{lepikhin2020gshard}, DSelect-k \cite{hazimeh2021dselect}, and Expert Choice Routing \cite{zhou2022expertchoice}, and leverage them on four different sequence recommendation models as backbones: Mamba4Rec \cite{liu2024mamba4rec}, TransAct \cite{xia2023transact}, TWIN \cite{chang2023twin}, and SDIM \cite{sdim}. These backbones cover both Transformer-based and Mamba-based architectures. From the multi-scale fusion perspective, we further compare MoS with long-sequence recommendation models that adopt similar fusion designs, including MIRRN \cite{xu2025mirrn}, AttenMixer \cite{atten-mixer}, and MiasRec \cite{MiasRec}.

\paragraph{Parameter Settings.}  Unless otherwise specified, we follow the default hyperparameter settings provided in the released code of the BARS Benchmark \cite{zhu2022bars}\footnote{\url{https://github.com/reczoo/LongCTR}}. We adopt the Adam optimizer~\cite{kingma2014adam} to train all models, with the initial learning rate tuned from $\{1\text{e-}4, 5\text{e-}4, 1\text{e-}3\}$. For MoS, the codebook dimension is searched from $\{16, 32, 64\}$, and a 2-layer MLP is used as $h(\cdot)$. The EMA update weight $\gamma$ is set to 0.999. For the window experts, we set the stride to 4 and the window size to 8. The multi-scale fusion weights are fixed at $\alpha_I = \alpha_W = 0.25$. For all MoE routers, the number of experts is set to 5, and the number of experts selected by the router $k$ is tuned from $\{1, 2\}$. All experiments are conducted on NVIDIA A100 GPUs.

\subsection{Experimental Results}

\begin{table*}[htbp]
\centering
\caption{Main evaluation of model utility on sequential recommendation. Higher AUC and GAUC values indicate better performance. Utility performance as well as the average ranking are reported on four sequential recommendation backbones, comparing MoS with three MoE baselines and the vanilla backbone (without MoE enhancement). Red font highlights \textcolor[HTML]{C00000}{the best performance}, while blue font denotes \textcolor[HTML]{3D74B6}{the second best}. Numbers in parentheses indicate the performance gain brought by MoE.}
\vspace{-2mm}
\resizebox{0.9\linewidth}{!}{
\begin{tabular}{cc|cccccc|cc}
\toprule
\multicolumn{2}{c|}{Dataset} & \multicolumn{2}{c}{MicroVideo} & \multicolumn{2}{c}{KuaiVideo} & \multicolumn{2}{c|}{Ebnerd} & \multicolumn{2}{c}{Rank} \\
\multicolumn{1}{l}{Backbone} & \multicolumn{1}{l|}{Method} & AUC & GAUC & AUC & GAUC & AUC & GAUC & AUC & GAUC \\
\midrule
    \multirow{5}{*}{\rotatebox{90}{Mamba4Rec}} & Vanilla & 66.55 & \textcolor[HTML]{3D74B6}{69.16} & 66.27 & \textcolor[HTML]{3D74B6}{66.26} & 63.08 & 62.68 & 4.00 & 2.67 \\
     & DSelect-k & \textcolor[HTML]{3D74B6}{68.40\tiny(+1.85)} & 68.94\tiny(-0.22) & \textcolor[HTML]{3D74B6}{66.48\tiny(+0.21)} & 66.22\tiny(-0.04) & \textcolor[HTML]{3D74B6}{63.59\tiny(+0.51)} & 62.50\tiny(-0.18) & 2.00 & 4.00 \\
     & GShard & 66.05\tiny(-0.50) & 68.15\tiny(-1.01) & 65.89\tiny(-0.38) & 65.74\tiny(-0.52) & 63.26\tiny(+0.18) & 63.79\tiny(+1.11) & 4.67 & 4.33 \\
     & Expert & 66.07\tiny(-0.48) & 69.04\tiny(-0.12) & 66.29\tiny(+0.02) & 66.06\tiny(-0.20) & 63.56\tiny(+0.48) & \textcolor[HTML]{3D74B6}{64.18\tiny(+1.50)} & 3.33 & 3.00 \\
     & \mtd{} & \textcolor[HTML]{C00000}{69.46\tiny(+2.91)} & \textcolor[HTML]{C00000}{70.01\tiny(+0.85)} & \textcolor[HTML]{C00000}{66.67\tiny(+0.40)} & \textcolor[HTML]{C00000}{66.67\tiny(+0.41)} & \textcolor[HTML]{C00000}{64.07\tiny(+0.99)} & \textcolor[HTML]{C00000}{64.54\tiny(+1.86)} & \textbf{1.00} & \textbf{1.00} \\
\midrule
    \multirow{5}{*}{\rotatebox{90}{TransAct}} & Vanilla & 70.58 & 69.90 & 67.67 & \textcolor[HTML]{3D74B6}{65.92} & \textcolor[HTML]{3D74B6}{70.89} & 70.25 & 3.33 & 3.67 \\
     & DSelect-k & 70.67\tiny(+0.09) & 69.92\tiny(+0.02) & \textcolor[HTML]{3D74B6}{67.87\tiny(+0.20)} & 65.38\tiny(-0.54) & 70.68\tiny(-0.21) & \textcolor[HTML]{3D74B6}{70.30\tiny(+0.05)} & 3.67 & 3.33 \\
     & GShard & 70.70\tiny(+0.12) & 70.02\tiny(+0.12) & 67.62\tiny(-0.05) & 65.42\tiny(-0.50) & 70.69\tiny(-0.20) & 70.08\tiny(-0.17) & 3.67 & 3.67 \\
     & Expert & \textcolor[HTML]{3D74B6}{70.85\tiny(+0.27)} & \textcolor[HTML]{3D74B6}{70.24\tiny(+0.34)} & 67.52\tiny(-0.15) & 65.28\tiny(-0.64) & 70.87\tiny(-0.02) & 70.27\tiny(+0.02) & 3.33 & 3.33 \\
     & \mtd{} & \textcolor[HTML]{C00000}{71.34\tiny(+0.76)} & \textcolor[HTML]{C00000}{70.51\tiny(+0.61)} & \textcolor[HTML]{C00000}{67.90\tiny(+0.23)} & \textcolor[HTML]{C00000}{68.10\tiny(+2.18)} & \textcolor[HTML]{C00000}{71.06\tiny(+0.17)} & \textcolor[HTML]{C00000}{70.71\tiny(+0.46)} & \textbf{1.00} & \textbf{1.00} \\
\midrule
    \multirow{5}{*}{\rotatebox{90}{TWIN}} & Vanilla & \textcolor[HTML]{3D74B6}{70.34} & 69.67 & \textcolor[HTML]{3D74B6}{69.36} & 66.60 & 69.85 & 69.54 & 2.33 & 3.67 \\
     & DSelect-k & 70.02\tiny(-0.32) & \textcolor[HTML]{3D74B6}{70.10\tiny(+0.43)} & 68.96\tiny(-0.40) & \textcolor[HTML]{3D74B6}{66.64\tiny(+0.04)} & 69.40\tiny(-0.45) & 68.99\tiny(-0.55) & 4.00 & 3.00 \\
     & GShard & 70.13\tiny(-0.21) & 69.74\tiny(+0.07) & 68.84\tiny(-0.52) & 66.62\tiny(+0.02) & \textcolor[HTML]{3D74B6}{70.10\tiny(+0.25)} & \textcolor[HTML]{3D74B6}{69.58\tiny(+0.04)} & 3.33 & 2.67 \\
     & Expert & 69.90\tiny(-0.44) & 69.56\tiny(-0.11) & 68.96\tiny(-0.40) & 66.34\tiny(-0.26) & 69.80\tiny(-0.05) & 69.33\tiny(-0.21) & 4.00 & 4.67 \\
     & \mtd{} & \textcolor[HTML]{C00000}{71.11\tiny(+0.77)} & \textcolor[HTML]{C00000}{70.26\tiny(+0.59)} & \textcolor[HTML]{C00000}{69.62\tiny(+0.26)} & \textcolor[HTML]{C00000}{66.89\tiny(+0.29)} & \textcolor[HTML]{C00000}{70.30\tiny(+0.45)} & \textcolor[HTML]{C00000}{69.81\tiny(+0.27)} & \textbf{1.00} & \textbf{1.00} \\
\midrule
    \multirow{5}{*}{\rotatebox{90}{SDIM}} & Vanilla & 70.50 & 69.62 & 68.56 & 66.61 & \textcolor[HTML]{C00000}{69.63} & 69.03 & 3.00 & 4.00 \\
     & DSelect-k & 69.65\tiny(-0.85) & 69.24\tiny(-0.38) & 68.78\tiny(+0.22) & 66.48\tiny(-0.13) & 69.47\tiny(-0.16) & 69.07\tiny(+0.04) & 4.00 & 4.33 \\
     & GShard & \textcolor[HTML]{3D74B6}{70.87\tiny(+0.37)} & 69.88\tiny(+0.26) & \textcolor[HTML]{3D74B6}{68.90\tiny(+0.34)} & 66.52\tiny(-0.09) & 69.47\tiny(-0.16) & 69.07\tiny(+0.04) & 2.33 & 3.33 \\
     & Expert & 70.30\tiny(-0.20) & \textcolor[HTML]{3D74B6}{70.05\tiny(+0.43)} & 68.85\tiny(+0.29) & \textcolor[HTML]{3D74B6}{66.66\tiny(+0.05)} & 69.40\tiny(-0.23) & \textcolor[HTML]{3D74B6}{69.13\tiny(+0.10)} & 4.00 & 2.00 \\
     & \mtd{} & \textcolor[HTML]{C00000}{71.36\tiny(+0.86)} & \textcolor[HTML]{C00000}{70.23\tiny(+0.61)} & \textcolor[HTML]{C00000}{69.38\tiny(+0.82)} & \textcolor[HTML]{C00000}{66.75\tiny(+0.14)} & \textcolor[HTML]{3D74B6}{69.49\tiny(-0.14)} & \textcolor[HTML]{C00000}{69.40\tiny(+0.37)} & \textbf{1.33} & \textbf{1.00} \\
\midrule
\end{tabular}
}
\label{tab:main_eval}
\end{table*}
\begin{table}[]
\centering
\caption{Performance comparison between \mtd{} and baselines employing similar multi-scale fusion architectures.}
\vspace{-2mm}
\resizebox{\linewidth}{!}{
\begin{tabular}{cc|cccccc}
\toprule
    \multicolumn{2}{c|}{\multirow{2}{*}{Method}} & \multicolumn{2}{c}{MicroVideo} & \multicolumn{2}{c}{KuaiVideo} & \multicolumn{2}{c}{Ebnerd} \\
    \multicolumn{2}{c}{} & AUC & GAUC  & AUC & GAUC  & AUC & GAUC  \\
\midrule
    \multirow{3}{*}{\rotatebox{90}{Baseline}}    & MIRRN        & 67.17      & 68.74 & 67.94     & 65.5  & 69.19  & 69.01 \\
    & MiasRec      & 68.67      & 67.83 & 65.05     & 64.70  & 55.95  & 55.45 \\
    & AttenMixer   & 68.97      & 68.89 & 66.95     & 64.18 & 62.37  & 62.29 \\
\midrule
    \multirow{4}{*}{\rotatebox{90}{\mtd{}}}         & TWIN         & 70.83      & \textcolor[HTML]{3D74B6}{70.26} & \textcolor[HTML]{C00000}{69.62}     & \textcolor[HTML]{3D74B6}{66.89} & 66.89  & \textcolor[HTML]{3D74B6}{69.81} \\
     & SDIM         & \textcolor[HTML]{C00000}{71.36}      & 70.23 & \textcolor[HTML]{3D74B6}{69.38}     & 66.75 & \textcolor[HTML]{3D74B6}{69.49}  & 69.49 \\
     & TransAct     & \textcolor[HTML]{3D74B6}{71.34}      & \textcolor[HTML]{C00000}{70.51} & 67.90      & \textcolor[HTML]{C00000}{68.10}  & \textcolor[HTML]{C00000}{71.06}  & \textcolor[HTML]{C00000}{70.71} \\
     & Mamba4Rec    & 69.46      & 70.01 & 66.67     & 66.67 & 64.07  & 64.54 \\
\bottomrule
\end{tabular}
}
\label{tab:compare_fusion}
\vspace{-3mm}
\end{table}

\paragraph{Main results.} 

The primary utility evaluation results are summarized in Table~\ref{tab:main_eval} and Table~\ref{tab:compare_fusion}. (1) Table~\ref{tab:main_eval} compares \mtd{} with other MoE methods, highlighting the superiority of its routing strategy. \textbf{Across diverse backbones, \mtd{} consistently delivers overall performance gains on both AUC and GAUC}, achieving an average improvement of 0.68\% in AUC and 0.72\% in GAUC. Specifically, \textbf{\mtd{} attains the top rank in AUC and GAUC across all datasets}, with average ranks of 1 / 1 / 1 / 1 for AUC and 1 / 1 / 1 / 1.33 for GAUC across the four backbones. In certain cases, the improvement is particularly pronounced. For instance, on the MicroVideo dataset, \mtd{} improves the performance of Mamba4Rec by 2.91\% in AUC and 0.85\% in GAUC.
(2) Table~\ref{tab:compare_fusion} further compares \mtd{} with sequential recommendation models that adopt similar fusion designs. \textbf{Across all datasets, models enhanced with \mtd{} consistently achieve both the first and second highest ranks, maintaining a clear performance margin over the remaining backbones.} On three datasets, the best-performing \mtd{} variants surpass the strongest baseline by 2.39 (1.62), 1.68 (2.60), and 1.87 (1.70) in AUC (GAUC), respectively. These results collectively demonstrate the superiority of \mtd{} in enhancing predictive performance. Importantly, this advantage is not tied to a specific architecture but generalizes robustly across heterogeneous backbones.

\paragraph{Efficiency Analysis}
\begin{figure}[htbp]
    \centering
    \vspace{-2mm}
    \includegraphics[width=\linewidth]{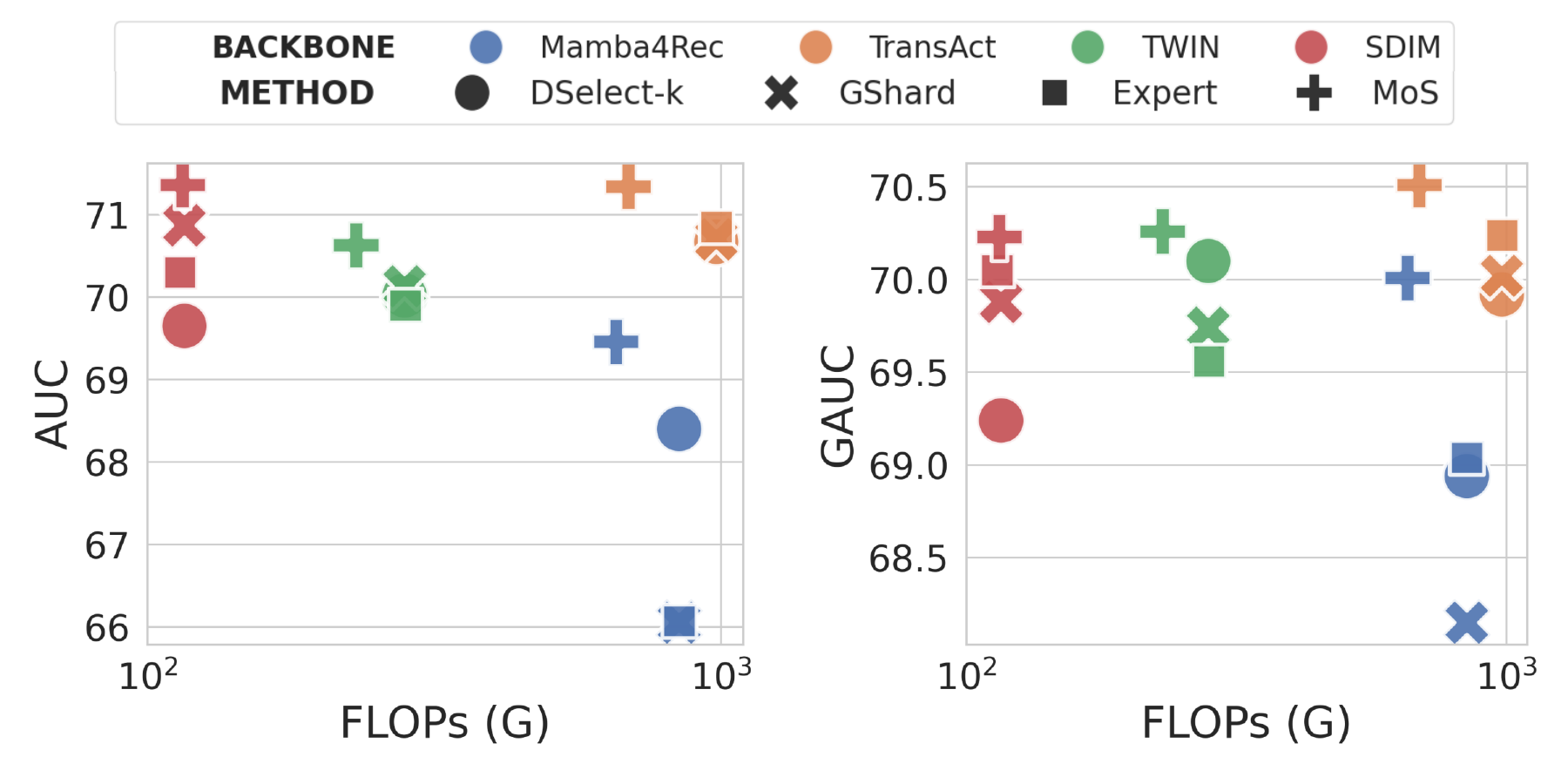}
    \vspace{-8mm}
    \caption{Trade-off between utility and efficiency. \mtd{} achieves superior utility–efficiency balance, appearing closer to the upper-left region for every backbone (color).}
    \label{fig:tradeoff}
    \vspace{-4mm}
\end{figure}

To evaluate the tradeoff between utility and efficiency of MoS, we compare it with all MoE methods using the same number of experts on the MicroVideo dataset across four backbones. The results are presented in Figure \ref{fig:tradeoff}, where colors denote different backbones and shapes correspond to different MoE methods. As shown, MoS achieves a highly competitive tradeoff, with all of its cross markers positioned near the upper-left corner across all backbones. More concretely, \textbf{MoS consistently delivers the best utility (both AUC and GAUC)}, outperforming all baselines on every backbone. At the same time, \textbf{compared with all MoE methods, MoS achieves the lowest FLOPs}, indicating that its subsequence extraction mechanism is highly efficient. These findings collectively demonstrate the superiority of MoS in balancing performance and efficiency.

\begin{figure}
    \centering
    \includegraphics[width=\linewidth]{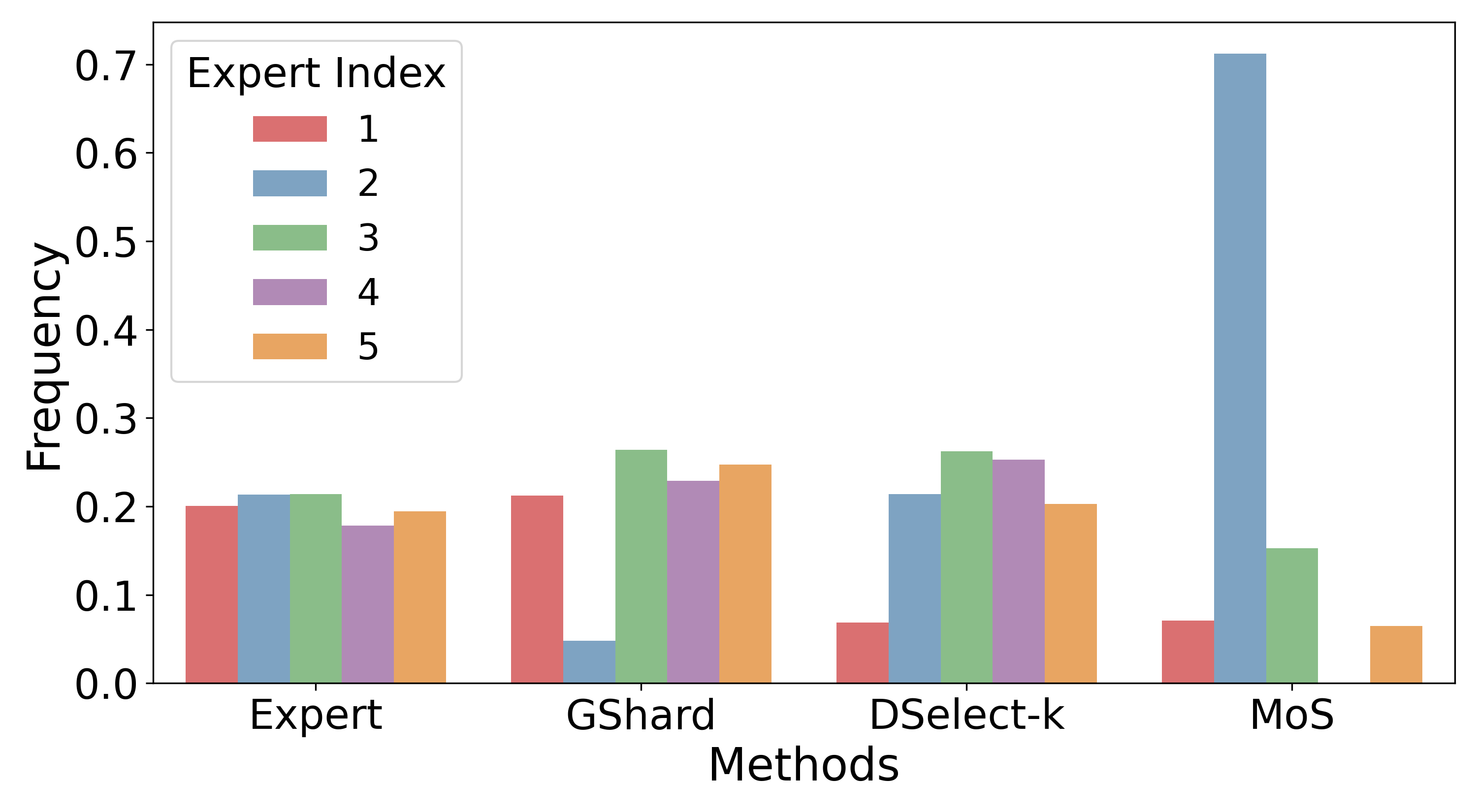}
    \vspace{-7mm}
    \caption{Dispatch behavior of routers for the most popular item. \mtd{} consistently routes the popular item to a fixed expert, indicating strong consistency of theme-aware routing and clear expert specialization compared to other routers.}
    \label{fig:routing}
    \vspace{-6mm}
\end{figure}

\paragraph{Routing Behavior Analysis.} We analyze the dispatch behavior of routers from all MoE variants to understand why \mtd{} consistently achieves superior performance on recommendation datasets. On the MicroVideo dataset, we select the most popular item, i.e., the item most frequently interacted with by all users, and then track how different MoE variants assign it to experts. Using the TransAct backbone, we record the expert selected for this item each time it appears and compute the dispatch frequency for each expert, as shown in Figure~\ref{fig:routing}. The results reveal that \mtd{} consistently routes the popular item to a fixed expert (here, expert \#2), demonstrating strong \textbf{theme-aware consistency and clear expert specialization}. In contrast, other MoE variants tend to distribute the same item almost uniformly across all experts. This uniformity indicates insufficient differentiation among experts, which in turn leads to degraded recommendation performance. More detailed explanation and analyses of routing behaviors are provided in Appendix~\ref{appdix:routing_behavior}.

\subsection{Ablation Studies}

\paragraph{Impact of multi-scale fusion.}
To examine the contributions of global, item, and window experts to prediction performance, we conduct an ablation study by varying $\alpha_I$ and $\alpha_W$ over ${0, 0.2, 0.4, 0.6, 0.8, 1}$. The corresponding AUC and GAUC results are shown in Figure~\ref{fig:alpha_ablation}. From these results, two clear conclusions can be drawn:

\textbf{(1) The multi-scale fusion mechanism effectively enhances recommendation performance.}
As shown in Figure~\ref{fig:alpha_ablation}, the heatmap corresponds to an upper-triangular region constrained by $\alpha_I + \alpha_W \leq 1$, where the vertices represent the use of single type of experts, the edges correspond to pairwise fusion of any two expert types, and the interior region denotes full fusion of all experts. It is evident that \textbf{full fusion consistently achieves the best results, while pairwise fusion outperforms using the single expert}. This observation suggests that although the global expert alone already provides strong representations for recommendations, the three types of experts capture complementary behavioral patterns and offer diverse analytical perspectives, thus leading to consistent performance gains when fused.
\begin{figure}[tbp]
    \centering
    \begin{subfigure}[b]{0.5\linewidth}
        \centering
        \resizebox{\linewidth}{!}{
        \includegraphics{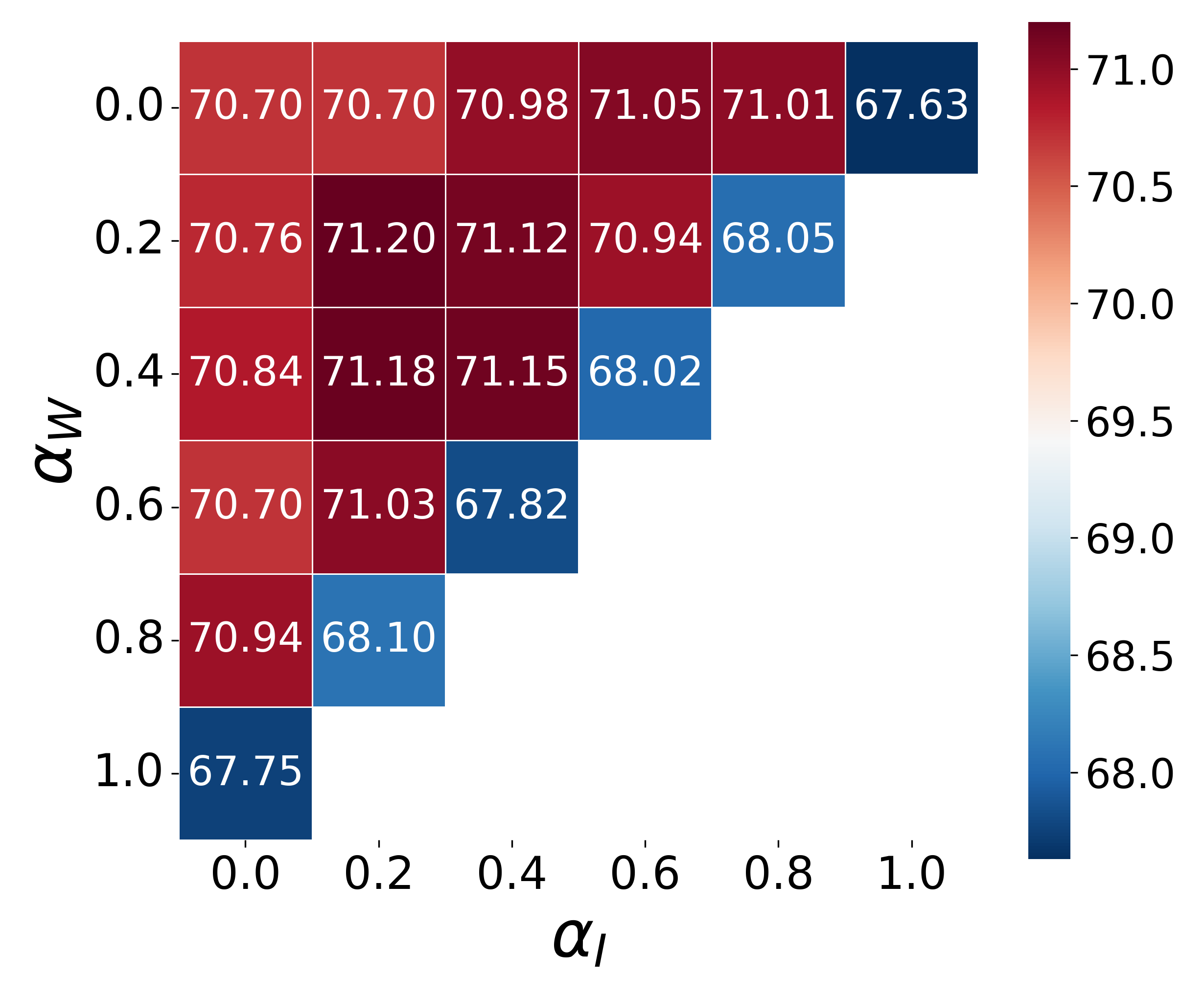}
        }
        \caption{Impact of $\alpha_I$ and $\alpha_W$ on AUC.}
        \vspace{-2mm}
    \end{subfigure}%
    \begin{subfigure}[b]{0.5\linewidth}
        \centering
        \resizebox{\linewidth}{!}{
        \includegraphics{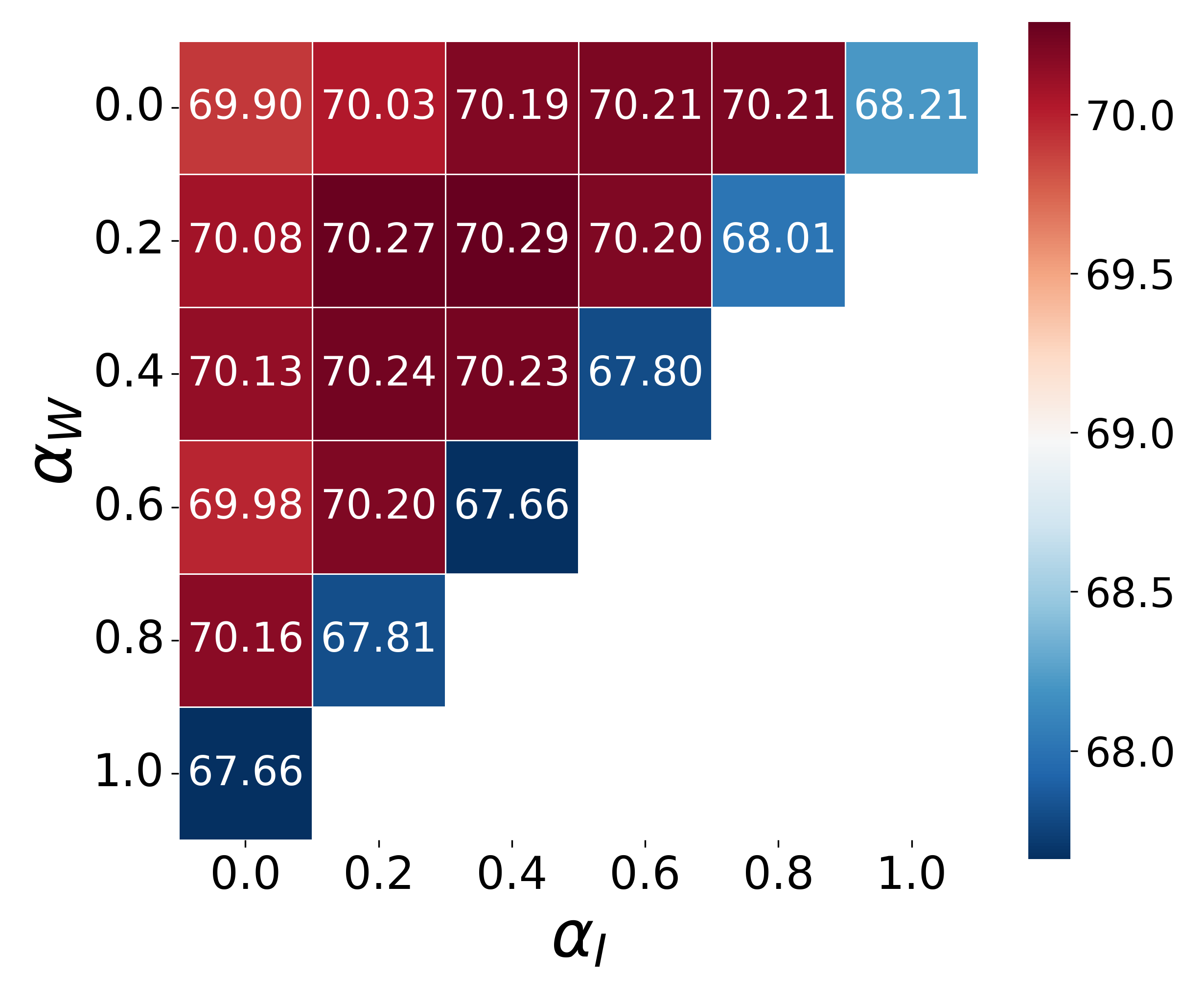}
        }
        \caption{Impact of $\alpha_I$ and $\alpha_W$ on GAUC.}
        \vspace{-2mm}
    \end{subfigure}
    \caption{Impact of $\alpha_I$ and $\alpha_W$ on model utility.}
    \label{fig:alpha_ablation}
    \vspace{-6mm}
\end{figure}

\begin{figure*}[ht]
    \centering
    \includegraphics[width=\linewidth]{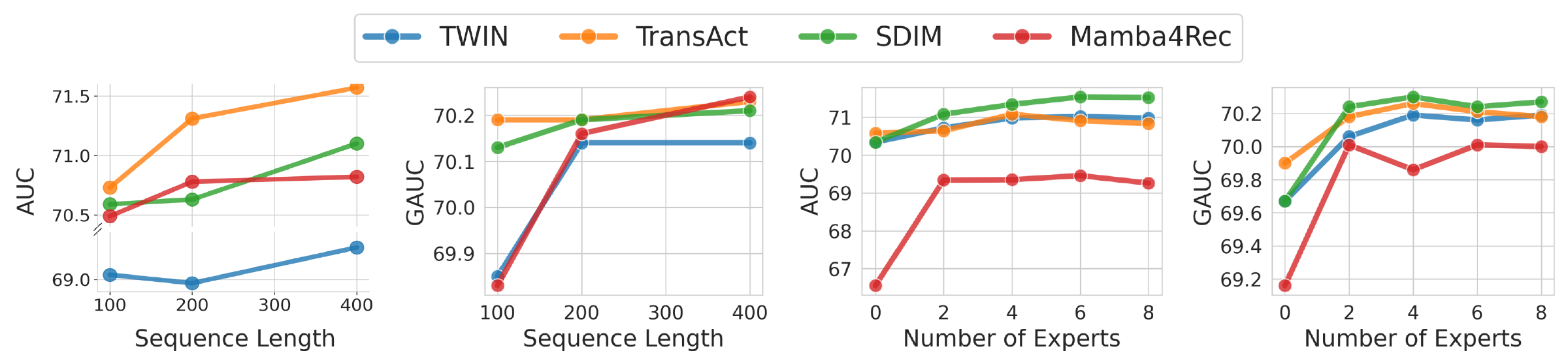}
    \vspace{-6mm}
    \caption{Scaling study of MoS. 
    MoS consistently enhances AUC/GAUC with increased sequence length and number of experts.}
    \vspace{-8pt}
    \label{fig:scaling}
\end{figure*}

\textbf{(2) The multi-scale fusion mechanism exhibits strong robustness.}
When all experts are fused, corresponding to the interior region of the triangle, the model performance remains consistently high without noticeable fluctuations. The AUC (GAUC) varies only within a narrow range from 70.94\% (70.20\%) to 71.20\% (70.27\%). The insensitivity of performance to $\alpha_I$ and $\alpha_W$ demonstrates the strong robustness of the multi-scale fusion mechanism.

\paragraph{Scaling on the sequence length.}
To assess the effectiveness of MoS in long-sequence recommendation, we vary the input sequence length on the MicroVideo dataset from 100 to 400 and set the number of experts as 5. As shown in Figure~\ref{fig:scaling}, both AUC and GAUC increase steadily with longer sequences across all four backbones. This trend confirms that MoS’s strategy of decomposing long sequences into shorter, theme-coherent subsequences effectively alleviates the impact of irrelevant information and training difficulty, resulting in superior and stable performance.

\vspace{-2mm}
\paragraph{Scaling on the number of experts.}
To examine the impact of the number of experts in MoS, we gradually increase the number of experts from 0 to 8 on the MicroVideo dataset with the sequence length of 200. As shown in Figure~\ref{fig:scaling}, model performance, including both AUC and GAUC, improves as the number of experts increases, but the improvement trend slows markedly once the number exceeds four. This phenomenon may be related to the underlying number of latent themes. When the number of experts is small, adding more experts effectively introduces additional themes, which helps to separate unrelated items or sessions and thereby enhances model performance. However, when the number of experts becomes sufficiently large, further increases lead to overly fine-grained distinctions between themes, causing many informative items to be split across different experts. As a result, the effective information within subsequences diminishes, producing performance saturation while increasing computational cost. These findings suggest that an appropriate number of experts yields beneficial effects for MoS.

\vspace{-2mm}
\section{Related Work}
\subsection{Click-Through Rate Prediction}
CTR~\citep{guo2017deepfm,xu2025multi,cheng2016wide} prediction is a fundamental task in online advertising, aiming to estimate the probability that a user clicks on a given item or ad.
A dominant paradigm is feature-interaction-based modeling, which typically follows the embedding–interaction–prediction pipeline. Early works~\citep{rendle2010factorization} capture low-order interactions, while Wide\&Deep~\citep{cheng2016wide} and DeepFM~\citep{guo2017deepfm} extend to higher-order with deep neural networks. More advanced designs include cross networks~\citep{,li2024fcn}, graph-based~\citep{ma2025graph} and attention-based models~\citep{xu2025multi, song2020towards}.

\vspace{-2mm}
\subsection{Sequential Recommendation}

Sequential recommendation~\citep{hidasi2015session, kang2018self} aims to predict the next item a user will interact with based on their historical behavior sequence. Advances in deep learning have driven the development of diverse models that capture sequential dependencies and uncover latent user interest from historical interactions. Representative approaches include Convolutional Neural Networks (CNNs)~\citep{Caser}, Recurrent Neural Networks (RNNs)~\citep{GRU4Rec,yue2024linear}, Transformer-based models~\citep{SASRec, BERT4Rec, FDSA}, and Graph Neural Networks (GNNs)~\citep{SRGNN, SURGE, DCRec, MSGIFSR}, all of which have been widely adopted to improve recommendation performance.
In addition, self-supervised learning has become a key paradigm in sequential recommendation~\citep{CL4SRec, ICLRec, Re4, DuoRec}, where contrastive objectives are applied to improve representation learning from unlabeled sequences. 
More recently, research on long-term user behavior modeling~\citep{liu2024mamba4rec,liu2024enhancing,xia2023transact,chang2023twin} has focused on improving efficiency and scalability under strict latency constraints. Representative directions include sampling-based hashing methods like SDIM~\citep{sdim}, architecture innovations such as Mamba4Rec~\citep{liu2024mamba4rec}, industrial hybrid solutions like TransAct~\citep{xia2023transact}, and life-long user modeling approaches such as TWIN~\citep{chang2023twin}, which advance the balance between effectiveness and efficiency in long-term sequential modeling. 
Our method provides a plug-and-play solution that can be integrated into existing long-sequence recommendation models.

\vspace{-2mm}
\section{Conclusion}

In this paper, we identify the session hopping phenomenon, which characterizes the stability, discontinuity, and recurrence of user interests—revealing the intrinsic challenges of long-sequence modeling. To address this issue, we propose \mtd{}, a framework that extracts multi-scale, theme-consistent subsequences for accurate prediction. \mtd{} introduces a \textit{theme-aware router} that learns latent user-interest themes via a special-designed codebook and extracts theme-specific subsequences accordingly. It further employs a \textit{multi-scale fusion mechanism} with three types of experts to capture global, short-term, and semantic representations for improved prediction. Extensive experiments show that \mtd{} consistently enhances existing long-sequence recommendation models, outperforming other MoE approaches while requiring fewer FLOPs, demonstrating an excellent balance between utility and efficiency.

\newpage
\bibliographystyle{ACM-Reference-Format}
\balance
\bibliography{contents/ref}

\appendix
\begin{figure*}
    \centering
    \includegraphics[width=0.9\linewidth]{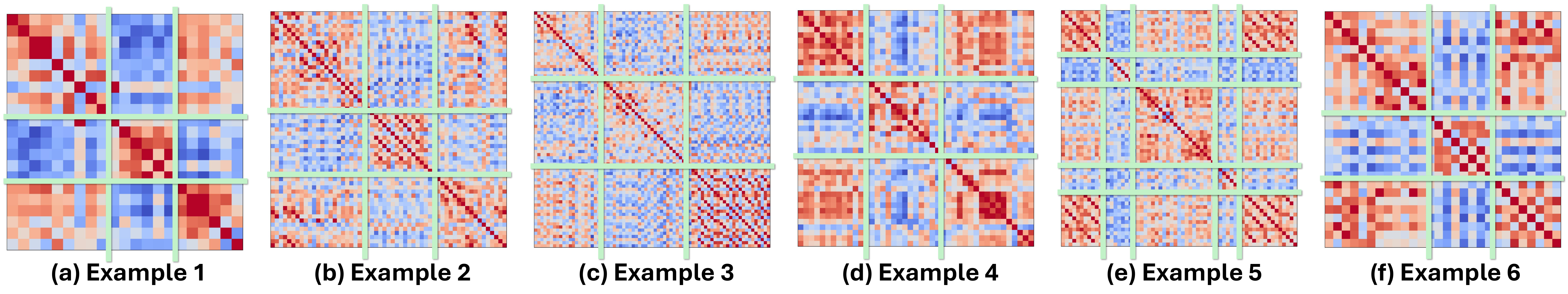}
    \vspace{-3mm}
    \caption{Examples of Session Hopping}
    \label{fig:examples_sessions}
    \vspace{-3mm}
\end{figure*}
\section*{Appendix}

\section{Examples of Session Hopping} \label{appdix:session_examples}

\begin{figure*}[htbp]
    \centering
    \begin{subfigure}[b]{0.4\linewidth}
        \centering
        \resizebox{\linewidth}{!}{
        \includegraphics{Figures/routing_behavior/frequency_distribution_by_method_frequent.png}
        }
        \caption{Router dispatch behavior for the most popular item.}
        \label{fig:popular_behavior_appdix}
    \end{subfigure}%
    \qquad
    \begin{subfigure}[b]{0.4\linewidth}
        \centering
        \resizebox{\linewidth}{!}{
        \includegraphics{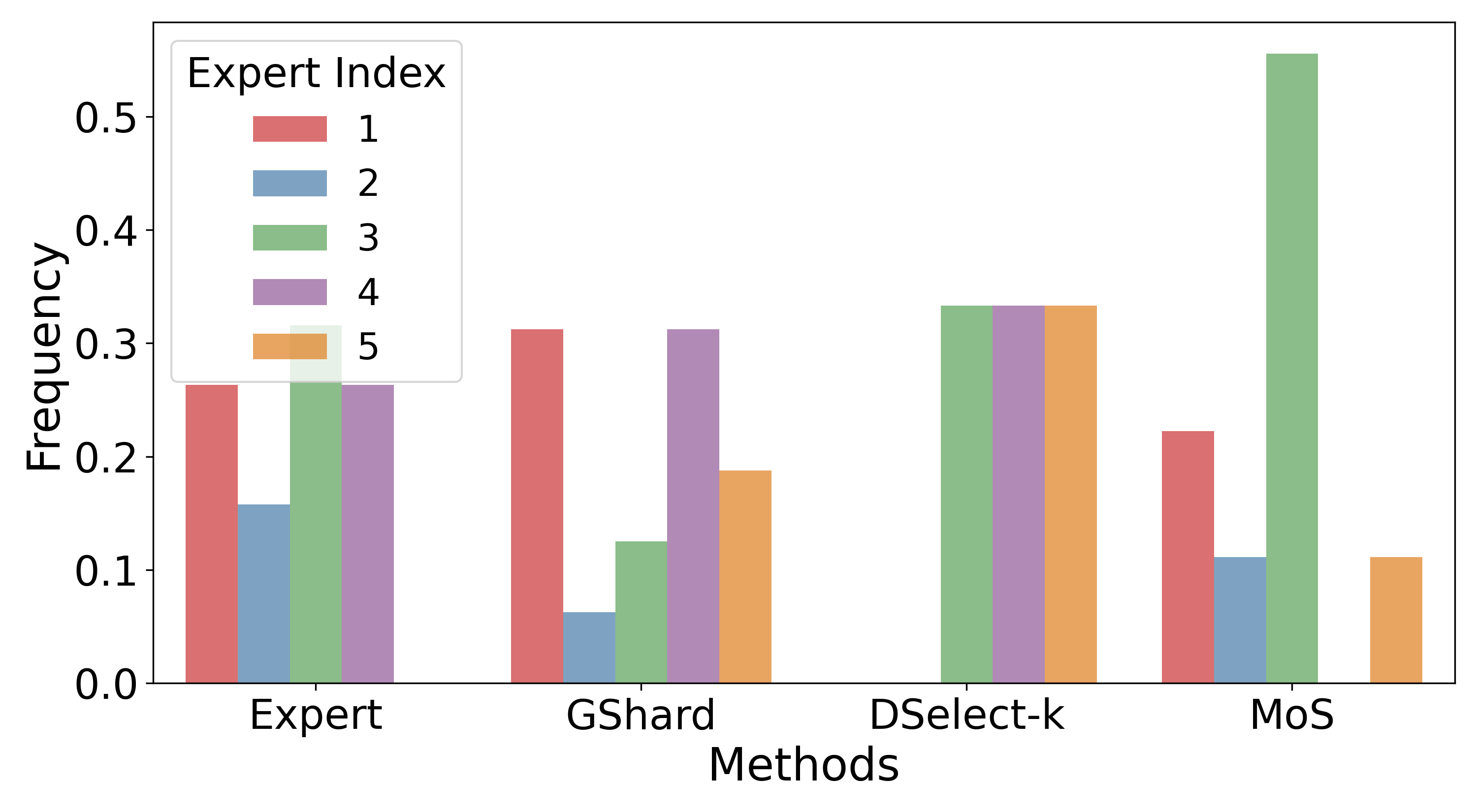}
        }
        \caption{Router dispatch behavior for the least popular item.}
        \label{fig:rare_behavior_appdix}
    \end{subfigure}
    \caption{Comparison between dispatch behavior of different MoE routers.}
    \label{fig:alpha_ablation_appdix}
\end{figure*}

To analyze user behavior patterns, we first use a pretrained embedding layer to obtain the embedding of each item in a user’s interaction history. The, we compute a self-similarity matrix of every user interaction sequence based on cosine similarity. A careful examination of these matrices reveals a widely recurring pattern, referred to in the main text as \textit{session hopping}. Representative cases are shown in Figure~\ref{fig:examples_sessions}. These examples consistently exhibit stability of interests within sessions, discontinuity across sessions, and the reappearance of interests over time, which together corroborate the prevalence of session hopping.

\section{Proof of Theme-aware Dispatch} \label{appdix:dispatch_proof}
\noindent\textbf{Setup.}
Let $\mW\in\mathbb{R}^{n\times d}$ with rows $\vw_i^\top$ and define $\va_i:=\vw_i/\|\vw_i\|_2$ for $i=1,\dots,n$.
For any input $\vx$, suppose $h(\vx)\neq \vzero$ and its direction $\vu(\vx):=h(\vx)/\|h(\vx)\|_2\in\mathbb{S}^{d-1}$.
Under the scoring in Eq.~(1), each coordinate of $H(\vx)$ equals the cosine similarity
\[
H_i(\vx)=\frac{\vw_i^\top h(\vx)}{\|\vw_i\|_2\,\|h(\vx)\|_2}
=\va_i^\top \vu(\vx),\qquad i=1,\dots,E.
\]
Let $H_{(1)}(\vx)\ge \cdots \ge H_{(n)}(\vx)$ be the sorted scores and let $S_k(\vx)$ be the index set of the top-$k$ entries of $H(\vx)$. Easy to prove, the support of $G(\vx)$ coincides with $S_k(\vx)$.

\begin{theorem}[Routing stability under cosine proximity]
\label{thm:topk-stability}
Fix $k\in\{1,\dots,n\}$ and two inputs $\vx_1,\vx_2$.
Assume: (i) the embedding directions are close in cosine,
\[
\rho:=\vu(\vx_1)^\top \vu(\vx_2)\ \ge\ 1-\delta \qquad\text{for some }\delta\in(0,1);
\]
(ii) there is a strictly positive top-$k$ margin at $\vx_1$,
\[
\gamma_k(\vx_1):=H_{(k)}(\vx_1)-H_{(k+1)}(\vx_1)>0.
\]
If $\delta<\gamma_k(\vx_1)^2/8$, then the selected experts are \emph{identical}, e.g., 
$
S_k(\vx_2)=S_k(\vx_1),
$
, hence $\supp G(\vx_2)=\supp G(\vx_1)$.
\end{theorem}

\begin{proof}
Set $\vu_1:=\vu(\vx_1)$ and $\vu_2:=\vu(\vx_2)$.
By unit-row normalization, each score is linear in $\vu$ with unit coefficient norm:
for all $i$, $H_i(\vx)=\va_i^\top \vu$ and $\|\va_i\|_2=1$.
Thus the coordinates of $H$ are $1$-Lipschitz in $\vu$:
\[
\big|H_i(\vx_1)-H_i(\vx_2)\big|
=\big|\va_i^\top(\vu_1-\vu_2)\big|
\le \|\vu_1-\vu_2\|_2
=: \varepsilon
\qquad\forall i.
\]
The cosine condition gives a bound on $\varepsilon$:
\[
\varepsilon
=\|\vu_1-\vu_2\|_2
=\sqrt{2-2\,\vu_1^\top\vu_2}
\le \sqrt{2\delta}.
\]
Pick any $i\in S_k(\vx_1)$ and any $j\notin S_k(\vx_1)$.
By the above coordinate-wise Lipschitz bound,
\[
H_i(\vx_2)\ \ge\ H_i(\vx_1)-\varepsilon,
\qquad
H_j(\vx_2)\ \le\ H_j(\vx_1)+\varepsilon,
\]
so
$
H_i(\vx_2)-H_j(\vx_2)
\ \ge\ \big(H_i(\vx_1)-H_j(\vx_1)\big)-2\varepsilon
\ \ge\ \gamma_k(\vx_1)-2\varepsilon.
$
If $\gamma_k(\vx_1)>2\varepsilon$, then $H_i(\vx_2)>H_j(\vx_2)$ for every such pair $(i,j)$, which preserves the entire top-$k$ ordering relation and hence the top-$k$ index set:
$S_k(\vx_2)=S_k(\vx_1)$.
Since $\varepsilon\le \sqrt{2\delta}$, the sufficient condition $\gamma_k(\vx_1)>2\varepsilon$ holds whenever
$
\delta<\frac{\gamma_k(\vx_1)^2}{8},
$
which proves the claim.
Finally, $\supp G(\vx)=S_k(\vx)$ by construction of $\KeepTopK$ and $\softmax$, hence the selected experts coincide as stated.
\end{proof}

\noindent\textbf{Remarks.}
In \mtd{}, $h(\cdot)$ is implemented as an MLP that satisfies Lipschitz continuity. The Lipschitz continuity of $h(\cdot)$ guarantees that small input perturbations induce small changes of the embedding direction and therefore of the cosine term in (ii).
The stability threshold in Theorem~\ref{thm:topk-stability} is explicit: the larger the top-$k$ margin $\gamma_k(\vx_1)$, the larger the admissible $\delta$.
Moreover, if any two items within a sequence exhibit high similarity, then by Theorem~\ref{thm:topk-stability}, all such items will be routed consistently to the same expert, thereby completing the proof of Proposition~\ref{prop:theme_dispatch}.

\section{Detailed Experimental Settings}
\subsection{Dataset Descriptions} \label{appdx:dataset_desc}

In this paper, we validate the effectiveness of \mtd{} on 3 different datasets, MicroVideo \cite{microvideo}, KuaiVideo \cite{kuaivideo}, and EBNeRD \cite{ebnerd}. The detailed information is listed below.

\begin{itemize}
    \item MicroVideo is provided by the THACIL work, which contains 12,737,617 interactions that 10,986 users have made on 1,704,880 micro-videos. The features include user id, item id, category, and the extracted image embedding vectors of cover images of micro-videos.
    \item KuaiVideo is released by the Kuaishou Competition in the China MM 2018 conference, which aims to predict users' click probabilities for new micro-videos. In this dataset, there are multiple types of interactions between users and micro-videos, such as "click", "not click", "like", and "follow". Particularly, "not click" means the user did not click the micro-video after previewing its thumbnail. 
    \item Ekstra Bladet News Recommendation Dataset (EBNeRD) encompasses data from over a million unique users and more than 37 million impression logs from Ekstra Bladet. It also includes a collection of over 125,000 Danish news articles, complete with titles, abstracts, bodies, and categories.
\end{itemize}

\subsection{Metric Descriptions}\label{appdix:metric_desc}
\paragraph{AUC (ROC AUC)}
AUC evaluates a binary scorer \(s(x)\) as a ranking metric. It equals the probability that a randomly drawn positive instance receives a higher score than a randomly drawn negative.  Let \(\mathcal{P}\) and \(\mathcal{N}\) be the sets of positive and negative examples with sizes \(n_{+}\) and \(n_{-}\). With scores \(s_i\) and standard tie handling, we have
\[
\mathrm{AUC}
= \frac{1}{n_{+} n_{-}}
\sum_{i\in \mathcal{P}} \sum_{j\in \mathcal{N}}
\Big[\mathbf{1}\{s_i > s_j\} + \tfrac{1}{2}\mathbf{1}\{s_i = s_j\}\Big].
\]
Equivalently, let \(R_i\) be the rank of \(s_i\) among all scores in ascending order and let the index set of positives be \(\mathcal{P}\). Then
\[
\mathrm{AUC}
= \frac{\sum_{i\in \mathcal{P}} R_i - \frac{n_{+}(n_{+}+1)}{2}}{n_{+} n_{-}}.
\]
This statistic matches the Wilcoxon Mann Whitney test scaled to the unit interval, which provides robustness under class imbalance.

\paragraph{GAUC (Group AUC)}
GAUC extends AUC to settings with natural groups such as per user or per query evaluation. The goal is to measure within group ranking quality and then aggregate across groups. For each group \(g\) that contains at least one positive and one negative, compute its AUC value \(\mathrm{AUC}_g\) as above.  The overall GAUC is a weighted average. Pair count weighting is common and unbiased for pairwise comparisons. Let \(n_g^{+}\) and \(n_g^{-}\) be the numbers of positives and negatives in group \(g\). GAUC could be expressed as:
\[
\mathrm{GAUC} = \frac{\sum_{g} w_g \,\mathrm{AUC}_g}{\sum_{g} w_g},
\qquad
w_g = n_g^{+} n_g^{-}.
\]

\paragraph{FLOPs}
FLOPs measure theoretical arithmetic workload for floating point computation. A single multiply or add counts as one floating point operation. FLOPs are hardware agnostic and are widely used to measure the efficiency of modern machine learning algorithm.

\section{Routing Behavior Analysis} \label{appdix:routing_behavior}

To better understand why \mtd{} outperforms other MoE methods, we conduct a fine-grained analysis of the router’s dispatch behavior at the item level. Specifically, on the MicroVideo dataset, many items are repeatedly selected by different users. We first filter out items selected fewer than eight times, and then identify two representative targets: the most frequently selected item and the least frequently selected item. For each target item, we record the expert assigned by the router at every selection and compute the frequency with which each expert is chosen. All MoE variants are configured with five experts. Figure~\ref{fig:popular_behavior_appdix} and Figure~\ref{fig:rare_behavior_appdix} illustrate the dispatch behavior for the most popular and least popular items, respectively.

From these results, several observations can be made.
(1) \mtd{} consistently routes the same item to a fixed expert, effectively leveraging the theme-aware specialization of experts. In contrast, other MoE methods tend to distribute items uniformly across experts, reflecting weak expert differentiation.
(2) When handling the most popular and the least popular items, \mtd{} relies on distinct experts. Specifically, expert \#2 is responsible for the most popular item, whereas expert \#3 dominates for the least popular one. Based on these two observations, we safely draw the conclusion that \textbf{the consistency of dispatch behavior and the specialization of experts are central to the superior performance of \mtd{}}.

\section{More Related Works}

\subsection{Mixture of Experts in Recommendation}
The Mixture-of-Experts architecture~\citep{jacobs1991adaptive} was originally introduced to improve supervised learning by combining multiple specialized models through a gating mechanism \cite{ai2025resmoe}. MoE has been widely adopted in recommendation systems~\citep{qin2020multitask,li2020video,bian2023multi,zhang2025hierarchical,zeng2025hierarchical,liu2024collaborative} particularly in multi-task learning scenarios. For instance, MMOE~\citep{ma2018modeling} introduces a multi-gate mechanism to balance multiple objectives, SNR~\citep{ma2019snr} improves the flexibility of parameter sharing via sub-network routing, and PLE~\citep{tang2020progressive} proposes progressive layered extraction to mitigate the seesaw effect in multi-task recommendation. MoE has also been further extended to sequential recommendation~\citep{bian2023multi} and cross-domain recommendation~\citep{hou2022towards,zhang2025cold}.
Despite these advances, traditional MoE methods route inputs to all experts, leading to high computational cost. Sparse MoE (SMoE) addresses this by activating only a subset of experts~\citep{MOE1,MOE2,switchtrans,pan2025fsmoe}, significantly improving efficiency while preserving model capacity. 

\subsection{Structured User Behavior Modeling}
In the era of big data~\cite{zeng2025harnessing,zeng2026subspace,li2025flow,wang2023networked,wang2023noisy,ban2021ee}, user recommendation data are collected and represented through a wide range of data modalities, including text~\cite{xu2025instructagent,wei2025cofirec,xu2024slmrec,liu2025selfelicit,wei2026agentic}, images~\cite{weitowards,ning2025graph4mm,lin2025moralise,liu2025climb,liu2025seeing}, time series~\cite{xu2024towards,ning2024information,lin2024backtime,lin2025cats,liang2025external,liu2025breaking}, and graph-structured data~\cite{zeng2025pave,yu2025joint,yu2025planetalign,chen2024masked,liu2024class}. Consequently, user behaviors exhibit rich structural patterns, such as temporal dynamics~\cite{yoo2024ensuring,yoo2025embracing,yoo2025generalizable,zeng2025interformer} and relational dependencies~\cite{zhang2024multi,wei2020fast,chen2024macro,qiu2025graph,li2023metadata}, making effective modeling of structured signals a fundamental challenge in recommendation.

First, Graph neural networks (GNNs)~\cite{velivckovic2017graph,kipf2016semi,bei2025graphs,lin2024bemap,xu2023node,fu2023natural,xu2024slog,xu2022graph,xu2022generalized} have been widely adopted to model such graph-structured user–item interactions. Homogeneous GNNs~\cite{he2020lightgcn,fan2019graph,chang2020bundle,yan2021bright,yan2021dynamic,yan2022dissecting,roach2020canon} are commonly used for collaborative filtering due to their low-pass filtering property, which facilitates learning smooth and similar user–item representations. In contrast, heterogeneous GNNs~\cite{shi2020heterogeneous1,wang2020disenhan,yan2024pacer,yan2024topological,yan2024thegcn} incorporate diverse side information, such as social relations, item attributes, and contextual links, yielding expressive representations. Furthermore, graph contrastive learning~\cite{liu2024simgcl,yu2023xsimgcl,wei2022augmentations,cai2023lightgcl,zheng2024pyg,ding2022data,yang2024simce} improves robustness by enforcing representation consistency under graph augmentations.

Second, recent research focuses on session-based recommendation~\cite{xu2019graphSession1,li2017session2,liu2023session3,qiao2025multi1session4,hidasi2015session,lee2025session5llm,wang2021session6graph}, which aims to capture both short-term and long-term temporal dependencies from sequential interaction data. Early approaches~\cite{GRU4Rec,li2017session2} predominantly adopt recurrent neural networks to capture sequential transition patterns within sessions, while subsequent methods~\cite{liu2018stamp,MiasRec,CORE} introduce attention mechanisms to better distinguish the main intent of a session from transient noise. More recently, graph-based session models~\cite{SRGNN,GCSAN,GCE-GNN,xu2019graphSession1} have been proposed to explicitly model item-to-item transitions by constructing session graphs, enabling the capture of complex local dependencies beyond simple sequential orderings. Extensions along this direction further exploit higher-order relations among items, such as hypergraph-based formulations~\cite{wang2021session6graph,su2023enhancing} that model item co-occurrence under multiple contextual windows, allowing richer semantic interactions to be encoded within a session. In parallel, recent studies~\cite{MSGIFSR,MiasRec,atten-mixer} explore enhanced session representations via disentangling short-term intent from broader contextual signals. Collectively, these session-based approaches demonstrate strong effectiveness in capturing short-term user preferences, yet they largely assume relatively homogeneous interest patterns within a session and remain challenged when user intents exhibit abrupt shifts or recur intermittently over longer interaction horizons.

\section{Ethical Use of Data and Informed Consent}

All experiments in this paper use only publicly available datasets from the LongCTR benchmark \footnote{\url{https://github.com/reczoo/LongCTR}}. We did not collect new data, contact users, or conduct any interventions. We complied with the licenses and terms of use specified by the benchmark and its constituent datasets. 
To support transparency, we will release our code, configuration files, and evaluation scripts upon publication.

\end{document}